\documentclass[american]{article}
\usepackage[top=1in, bottom=1in, left=1.5in, right=1.5in]{geometry}

\usepackage[T1]{fontenc}
\usepackage[american]{babel}
\usepackage{amsmath,amssymb,amsthm}
\usepackage{slashed}
\usepackage{nicematrix}
\usepackage{setspace}
\usepackage[authoryear]{natbib}
\usepackage{microtype}
\usepackage[colorlinks=true,linkcolor=blue,citecolor=blue,urlcolor=blue]{hyperref}

\theoremstyle{plain}
\newtheorem*{thm*}{Theorem}
\newtheorem{lem}{Lemma}
\newtheorem{cor}{Corollary}
\newtheorem{thm}{Theorem}
\newtheorem{prop}{Proposition}

\theoremstyle{definition}
\newtheorem{defn}{Definition}
\newtheorem{example}{Example}

\theoremstyle{remark}
\newtheorem{rem}{Remark}

\newcommand{\supp}{\operatorname{supp}}

\begin{document}
\title{Continuity of stable marriage correspondences\thanks{Haluk Ergin was a named author on an earlier version of this paper. We thank Haluk Ergin and Chris Shannon for guidance. Mistakes are solely the named authors'.\\
\textbf{Sandholtz:} willsandholtz@berkeley.edu. This work was begun when the author was a PhD student at UC Berkeley.\\
\textbf{Tai:} andrew.a.tai@gmail.com. US Department of War. Disclaimer: The views expressed are those of the author and do not reflect the official policy or position of the Department of War or the U.S. Government.}}
\author{Will Sandholtz \and Andrew Tai}
\date{September 2026}
\maketitle
\begin{abstract}
\noindent We study a continuum version of the stable marriage problem, in which each agent has a variable mass to be paired to the opposite side. An agent's mass is interpretable as the divisible time of a single atomistic agent, or as a mass of infinitesimal agents of a particular type. Our object of interest is the mapping between the vector of masses to the set of stable matches. As in the canonical stable marriage model, the set of stable matches can be multi-valued, so this mapping is a correspondence. This correspondence is continuous if and only if the stable match is unique, implying that small changes in the masses may cause discontinuous jumps in the set of stable matches.
\end{abstract}

\section{Introduction}

The stable marriage problem of \citet{GS62} is the canonical two-sided matching market. In their model, two sides of discrete agents seek to match with each other in a way that is \emph{stable}, in the sense that no man and woman who are not paired would be better off by marrying each other. \citet{BB02} extend this to a \emph{stable allocation} problem between ``workers'' and ``employers''. Each agent remains discrete, but there is a divisible mass of each agent. The goal is to allocate the time of workers to employers in a stable way, similarly defined---no mass of a pair of agents would break off and re-match with each other. The mass of an agent can be alternatively interpreted as infinitesimal individuals of a given type.

It is a well-known informal fact of the Gale and Shapley model that small changes to the market composition (e.g. addition or removal of agents) can result in dramatic changes to the set of stable matches.\footnote{For a formal result of this flavor, see \citet{AKL16}.} We extend and formalize this result to the stable allocation setting, taking advantage of the setting's additional topological structure. Like in the discrete case, the set of stable matches may be multi-valued. Thus the mapping from the vector of type masses (``populations'') to the set of stable matches is a set-valued function (\textbf{correspondence}). Our main result is that this correspondence is continuous if and only if it is single-valued; that is, when the stable matching is unique. An informal statement is below.
\begin{thm*}
The stable matching correspondence $f(\cdot)$ is continuous (upper and lower hemicontinuous) at a population vector $p$ if and only if $f(p)$ is single-valued.
\end{thm*}
Before we present the model and technical results, we note some applications of the setting and our results.

\paragraph*{Gig platforms.}
On gig platforms, workers and consumers are plentiful and can be modeled symmetrically. The platform seeks to match the time of workers to the needs of consumers.\footnote{This contrasts with the model in \citet{AL16}, where one side is plentiful (students), and the other is finite (schools).}

\paragraph*{Stable allocation of tasks.}
One side is workers, and the other is employers (or projects). Rather than having a mass of agents, each discrete agent has a mass of available or required time.

\paragraph*{Marriage markets.}
\citet{HWZ23} build a similar continuum marriage market with four discrete types per side.\footnote{Their model neither contains nor is contained by ours; we refer readers to their paper for details.} They note that small demographic shifts can drastically alter results. Our main theorem isolates a source of change; since the stable matching correspondence is not generally continuous, even infinitesimal changes to the population vector can cause large shifts in the set of stable matches.

\paragraph*{Negative implications.}
In the model we study, $f(p)$ may be multi-valued. This introduces two sources of difficulties for comparative statics: selection of a stable match inside $f(p)$ matters; and changes to $p$ can cause discontinuous changes in $f(p)$.

\paragraph*{Positive implications.}
There are special cases where $f(p)$ is single-valued. \citet{BB02} note a sufficient condition for a unique stable matching, called ``strong nondegeneracy''. Informally, a market is strongly nondegenerate if neither it nor any proper subset has equal mass of the two sides. If the populations are drawn from real numbers (rather than integers), uniqueness and thus continuity is the typical case. We formalize this in Section~\ref{subsec:Uniqueness}.

Another immediate case is where at least one side shares identical preferences; e.g. all the men types have the same preferences over the women. In these cases, $f(p)$ is continuous.

\vspace{\baselineskip}
A number of other papers deal with two-sided matching in continuum settings, including \citet{CL24}, \citet{CKK19}, \citet{GK21}, and \citet{Menzel15}.

Section 2 formally describes the model. Section 3 describes useful structure of the set of stable matches. Section 4 contains the main theorem. Section 5 describes an implication of our main result to the approximation of large discrete markets with continuum markets.

\section{Model}

We study a two-sided marriage market with discrete agent types, which determine preferences, and masses of each type. The model was first noted in \citet{BB02}.\footnote{In particular, their unconstrained model. The main difference is notational---the unmatched option $\slashed{o}$ is not explicitly encoded; agents are instead allowed to be undermatched relative to their mass.}

There is a finite set $M$ of man types, each type with a strict preference order $\succ_{m}$ over $W\cup\{\slashed{o}\}$. Likewise, a finite set $W$ of women types, each with a strict preference order $\succ_{w}$ over $M\cup\{\slashed{o}\}$. We use $\slashed{o}$ to represent being unmatched, and suppress the set brackets in the future. There is a vector of positive masses of each type (\textbf{population}) $p\in\mathcal{P}=\mathbb{R}_{++}^{M\cup W}$.

An \textbf{aggregate matching} is a matrix $Q\in\mathcal{Q}=\mathbb{R}_{+}^{(M\cup\slashed{o})\times(W\cup\slashed{o})}$, where $Q_{m,w}$ represents the mass of type $m$ matched to type $w$. $Q_{\slashed{o},\slashed{o}}=0$ by convention. An aggregate matching is \textbf{feasible} if $\sum_{w\in W\cup\slashed{o}}Q_{m,w}=p_{m}$ for all $m\in M$ and $\sum_{m\in M\cup\slashed{o}}Q_{m,w}=p_{w}$ for all $w\in W$. Given $Q$, we denote the support of $Q$ as $\supp(Q):=\{(m,w):Q_{m,w}>0\}$.
\begin{defn}
A pair $(m,w)$ \textbf{blocks} $Q$ if there exist $(m',w')\in(M\cup\slashed{o})\times(W\cup\slashed{o})$ such that
\[
Q_{m,w'}>0,\quad Q_{m',w}>0
\]
and
\[
w\succ_{m}w',\quad m\succ_{w}m'.
\]
$Q$ is \textbf{pairwise stable (PS)} if there does not exist any blocking pair. $Q$ is \textbf{individually rational (IR)} if, whenever $Q_{m,w}>0$ for $m\in M,w\in W$, then $w\succeq_{m}\slashed{o}$ and $m\succeq_{w}\slashed{o}$. $Q$ is \textbf{stable} if it is PS and IR.
\end{defn}
\begin{example}
\label{exa:Ex}
Let a continuum market be $N=\{m_{1},m_{2},w_{1},w_{2}\}$ with preferences

\begin{table}[htbp]
\centering
\begin{tabular}{cccc}
$m_{1}$ & $m_{2}$ & $w_{1}$ & $w_{2}$\\
\hline
$w_{1}$ & $w_{2}$ & $m_{2}$ & $m_{1}$\\
$w_{2}$ & $w_{1}$ & $m_{1}$ & $m_{2}$\\
$\slashed{o}$ & $\slashed{o}$ & $\slashed{o}$ & $\slashed{o}$\\
\end{tabular}
\end{table}

and population $p=(1.5,1,1,1)$. The unique stable match is

\[
\begin{bNiceMatrix}[first-row,first-col]
            & w_1 & w_2 & \varnothing \\
m_1         & 0   & 1   & 0.5        \\
m_2         & 1   & 0   & 0          \\
\varnothing & 0   & 0   & 0
\end{bNiceMatrix}.
\]
\end{example}

We are interested in the stable matching correspondence $f:\mathcal{P}\rightrightarrows\mathcal{Q}$ defined as
\[
f(p)=\{Q:Q\text{ is stable and feasible at }p\}.
\]

We recount the definitions of hemicontinuity.
\begin{defn}
Let $f$ be compact-valued. Then $f$ is \textbf{upper hemicontinuous (uhc)} at $p\in\mathcal{P}$ if and only if for every sequence $p^{n}\rightarrow p$ in $\mathcal{P}$ and every sequence $Q^{n}\in f(p^{n})$ such that $Q^{n}\rightarrow Q$, we have $Q\in f(p)$.
\end{defn}

\begin{defn}
The correspondence $f$ is \textbf{lower hemicontinuous (lhc)} at $p\in\mathcal{P}$ if and only if for every sequence $p^{n}\rightarrow p$ in $\mathcal{P}$ and every $Q\in f(p)$, there exists a subsequence $p^{m}$ and sequence $Q^{m}\in f(p^{m})$ such that $Q^{m}\rightarrow Q$.
\end{defn}
\begin{rem}
These are not the typical primitive definitions of uhc and lhc, but they are equivalent given the compact-valuedness of $f(\cdot)$. We will show that $f(p)$ is closed; boundedness is immediate.
\end{rem}

\section{Structure of \texorpdfstring{$f(p)$}{f(p)}}

Many results on the structure of stable matches are analogous to standard results for the original Gale and Shapley setting; we refer the reader to \citet{RS92}.

Using first-order stochastic dominance, we extend the preference $\succeq_{i}$ of each type $i$ to a partial order $\geq_{i}$ over matchings as follows. Let $Q^{1},Q^{2}\in f(p)$ be two matchings.
\begin{itemize}
\item $\forall m\in M:$ Let $Q_{m}^{1}\geq_{m}Q_{m}^{2}$ if for all $\bar{w}\in W\cup\{\emptyset\}$:
\[
\sum_{w\in W\cup\{\emptyset\}:w\succeq_{m}\bar{w}}Q_{m,w}^{1}\geq\sum_{w\in W\cup\{\emptyset\}:w\succeq_{m}\bar{w}}Q_{m,w}^{2}
\]
\item $\forall w\in W:$ Let $Q_{w}^{1}\geq_{w}Q_{w}^{2}$ if for all $\bar{m}\in M\cup\{\emptyset\}$:
\[
\sum_{m\in M\cup\{\emptyset\}:m\succeq_{w}\bar{m}}Q_{m,w}^{1}\geq\sum_{m\in M\cup\{\emptyset\}:m\succeq_{w}\bar{m}}Q_{m,w}^{2}
\]
\end{itemize}
Now we define a partial order for each side over matchings.
\begin{itemize}
\item Let $Q^{1}\geq_{M}Q^{2}$ if $\forall m\in M:Q_{m}^{1}\geq_{m}Q_{m}^{2}$
\item Let $Q^{1}\geq_{W}Q^{2}$ if $\forall w\in W:Q_{w}^{1}\geq_{w}Q_{w}^{2}$
\end{itemize}
The following are due to \citet{BB02}.
\begin{lem}[Existence. \citealp{BB02}. Corollary 1.]
For any $p$, $f(p)$ is nonempty.
\end{lem}

\begin{lem}[Opposing interests. \citealp{BB02}. Theorems 5, 6, and 10.]
Let $Q^{1},Q^{2}\in f(p)$. Then $Q^{1}\geq_{M}Q^{2}$ if and only if $Q^{1}\leq_{W}Q^{2}$. Further, there exists a unique $\ge_{M}$-maximal matching, denoted $Q^{M}$, and a unique $\ge_{W}$-maximal matching $Q^{W}$.
\end{lem}

\begin{lem}[Rural hospitals. \citealp{BB02}. Lemmas 8--9.]
\label{lem:rural_hosp}
In any $Q^{1},Q^{2}\in f(p)$,
\begin{align*}
Q_{m,\slashed{o}}^{1} & =Q_{m,\slashed{o}}^{2} & \forall m\\
Q_{\slashed{o},w}^{1} & =Q_{\slashed{o},w}^{2} & \forall w.
\end{align*}
Further, if $m$ has $Q_{m,\slashed{o}}>0$ in some $Q\in f(p)$, then $Q_{m,w}=Q_{m,w}^{'}$ for all $Q'\in f(p)$. The analogous is true for any $w$.
\end{lem}

We now present original results that are useful for proving the main result. First we introduce an alternative characterization of $f(p)$. The following is straightforward from the definition of stability.
\begin{lem}
\label{lem:support}
Fix $(M,W,\succeq,p)$. Let $Q,Q'$ be two feasible matchings such that $\supp(Q')\subseteq\supp(Q)$. If $Q\in f(p)$ then $Q'\in f(p)$. Further, if $\supp(Q')=\supp(Q)$, then $Q\in f(p)\iff Q'\in f(p)$.
\end{lem}
\begin{proof}
Omitted.
\end{proof}
The above also suffices to show that $f(\cdot)$ is homogeneous of degree one.
\begin{lem}
\label{lem:homog}
Let $\lambda>0$. Then $f(\lambda p)=\lambda f(p)$.
\end{lem}
\begin{proof}
Omitted.
\end{proof}
In words, two feasible matchings with the same support have equivalent stability, regardless of the actual masses $Q_{m,w}$. Given a stable match $Q$, a feasible match with a smaller support is also stable. This is because by definition, there are fewer conditions to check for stability. In particular, if $Q\in f(p)$, then the set
\[
\left\{ Q':\begin{array}{cc}
Q'_{m,w}\geq0 & \text{if }(m,w)\in\supp(Q)\\
Q'_{m,w}=0 & \text{otherwise}
\end{array}\right\}
\]
is stable. Observe that this set is closed and bounded (since each $Q'$ respects the population $p$). Since there are finite possible supports of stable matchings given a population $p$, taking the union across supports shows that $f(p)$ is compact.
\begin{cor}
\label{cor:compact}
$f(p)$ is compact.
\end{cor}
\begin{proof}
Let
\[
\mathcal{S}:=\bigcup_{Q\in f(p)}\{\supp(Q)\}
\]
which is finite, since $\mathcal{S}\subseteq2^{(M\cup\slashed{o})\times(W\cup\slashed{o})}$. Then we can write
\[
f(p)=\bigcup_{S\in\mathcal{S}}\left\{ Q':\begin{array}{cc}
Q'_{m,w}\geq0 & \text{if }(m,w)\in S\\
Q'_{m,w}=0 & \text{otherwise}
\end{array}\right\} ,
\]
which is a finite union of closed and bounded sets.
\end{proof}
The next example illustrates a case of multiple stable matches, as well as some of the above properties.
\begin{example}
\label{exa:2}
Continuing from Example \ref{exa:Ex}, consider the same types and preferences with population $p=(1,1,1,1)$. The set of stable matches is

\[
Q^{t} =
\begin{bNiceMatrix}[first-row,first-col]
            & w_1 & w_2 & \varnothing \\
m_1         & t   & 1-t   & 0        \\
m_2         & 1-t  & t   & 0          \\
\varnothing & 0   & 0   & 0
\end{bNiceMatrix}
\]
for $t\in[0,1]$.

Observe that for $t\geq t'$, we have $Q^{t}\geq_{M}Q^{t'}$ and $Q^{t}\leq_{W}Q^{t'}$.
\end{example}

\section{Continuity}

We now turn to the main results. The first result is that $f$ is uhc everywhere and thus continuous when it is single-valued.
\begin{thm}
\label{thm:uhc}
The stable matching correspondence $f:\mathcal{P}\rightrightarrows\mathcal{Q}$ is upper hemicontinuous everywhere. If $|f(p)|=1$, then $f(\cdot)$ is continuous at $p$.
\end{thm}
\begin{proof}
Let $p^{n}\rightarrow p$ be a convergent sequence in $\mathcal{P}$. Let $Q^{n}\in f(p^{n})$ be a convergent sequence such that $Q^{n}\rightarrow Q$. We wish to show $Q\in f(p)$. That is, we show $Q$ is feasible, IR, and PS at $p$.
\begin{itemize}
\item \textbf{Feasibility.} Since row and column sums of $Q$ are continuous functions of $p$, feasibility is preserved in the limit.
\item \textbf{IR.} Suppose there is $(m,w)$ such that $Q_{m,w}>0$ but $\slashed{o}\succ_{m}w$. Then for large enough $n$, $Q_{m,w}^{n}>0$, else $Q^{n}\not\rightarrow Q$, a contradiction to stability of $Q^{n}$. The $w$ case is symmetric.
\item \textbf{PS.} Suppose $(m,w)$ is a blocking pair for $Q$. Then there exists $(m',w')$ such that
\[
Q_{m,w'}>0,\quad Q_{m',w}>0,\quad w\succ_{m}w',\quad m\succ_{w}m'
\]
But then for large enough $n$, $Q_{m,w'}^{n}>0$ and $Q_{m',w}^{n}>0$, contradicting stability of $Q^{n}$.
\end{itemize}
When $|f(p)|=1$ and $f(\cdot)$ is compact-valued, equivalence of uhc to continuity is standard. Let $p^{n}\rightarrow p$, and choose any $Q^{n}\in f(p^{n})$. Since $p^{n}$ is bounded, the corresponding matchings are bounded, and $Q^{n}$ has a convergent subsequence. By uhc, its limit must lie in $f(p)$. Since $f(p)=\{Q\}$, every limit of a convergent subsequence is $Q$. Hence $Q^{n}\rightarrow Q$, as desired.
\end{proof}
The second and more difficult result is that $f$ is lhc at $p$ if and only if $f(p)$ is single-valued.
\begin{thm}
\label{thm:lhc}
The stable matching correspondence $f:\mathcal{P}\rightrightarrows\mathcal{Q}$ is lower hemicontinuous at $p$ if and only if $|f(p)|=1$.
\end{thm}
The ``if'' direction is a corollary of Theorem \ref{thm:uhc}. The remainder is in the following subsection.

\subsection{Proof of Theorem \ref{thm:lhc}}

We first introduce a graph construction. Fix $(M,W,\succeq,p)$ and let $Q\in f(p)$. Define the graph $G(Q)=(N,E)$ with nodes
\[
N=M\cup W\cup\slashed{o}
\]
and edges
\[
(m,w)\in E\iff(m,w)\in\supp(Q).
\]
By Lemma \ref{lem:support}, any feasible matrix $Q'$ such that $G(Q)=G(Q')$ is also stable.

Supposing $f(p)$ is single-valued, so $f(p)=\{Q\}$, the graph $G(Q)$ is acyclic; the converse is also true. (The converse is the important direction for our proof.)
\begin{lem}
\label{lem:acyclic}
If $f(p)=\{Q\}$, then $G(Q)$ is acyclic.
\end{lem}
\begin{proof}
Toward a contradiction, suppose $G(Q)$ contains a cycle, denoted
\[
\Gamma=\{(m_{1},w_{1}),(m_{2},w_{1}),(m_{2},w_{2}),\ldots,(m_{K},w_{K}),(m_{1},w_{K})\}
\]
where each edge $\gamma\in\Gamma$ has $Q_{\gamma}>0$. Let
\[
D=\sum_{k=1}^{K}(e_{m_{k},w_{k}}-e_{m_{k+1},w_{k}})
\]
where $e_{m_{k},w_{k}}$ is the basis vector in $\mathcal{Q}$ corresponding to $(m_{k},w_{k})$. That is, $D$ adds mass to odd edges and subtracts equivalent mass from even edges. For $\alpha\in[0,\min_{\gamma\in\Gamma}Q_{\gamma})$, the matching
\[
Q^{\alpha}=Q+\alpha D
\]
is feasible and has the same support $\supp(Q^{\alpha})=\supp(Q)$. Thus $Q^{\alpha}\in f(p)$, a contradiction.
\end{proof}
Supposing $f(p)$ is multi-valued, the following result identifies a cycle along which mass can flow while remaining stable.

First, we describe \citet{BB02}'s man-improving algorithm. Starting from any stable matching, the algorithm either certifies man type $m$ as optimal, or finds a directed cycle that can improve $m$. The usefulness of the algorithm (for our purpose) is that it generates a convex set of stable matches along a cycle; the welfare aspect and exact construction is unimportant for us. Readers may safely skip to Lemma \ref{lem:cycle}, which contains the needed implication for our proof.

Given a stable matching $Q$ and mutually acceptable pair $(m,w)$, call $(m,w)$ man-stable if
\[
\sum_{w'\succeq_{m}w}Q_{m,w'}=p_{m}
\]
and woman-stable if
\[
\sum_{m'\succeq_{w}m}Q_{m',w}=p_{w}
\]
The algorithm constructs a chain of men as follows:
\begin{enumerate}
\item Let $M^{*}$ be the subset of men whose allocations are already man-optimal. That is, $Q_{m}=Q_{m}^{M}$ for $m\in M^{*}$. Initialize $M^{*}$ with all $m\in M$ such that $Q_{m,\slashed{o}}>0$ (by \ref{lem:rural_hosp}). If $M^{*}=M$, stop.

Otherwise, choose $m_{0}\notin M^{*}$, and let $w_{0}$ be $m_{0}$'s worst realized partner in $Q$. Call $(m_{0},w_{0})$ even; we will remove mass from even cells.
\item At an even cell $(m_{t},w_{t})$, consider $w_{t}$'s preference ranking. Consider men $m\notin M^{*}$ such that $(m,w_{t})$ is not man-stable; so $m_{t}\succ_{w_{t}}m$. Choose the one most preferred by $w_{t}$, and call him $m_{t+1}$. (Existence is treated below.) Thus
\[
m_{t}\succ_{w_{t}}m_{t+1}
\]
and
\[
\sum_{w'\succeq_{m_{t+1}}w_{t}}Q_{m_{t+1},w'}<p_{m_{t+1}}
\]
The cell $(m_{t+1},w_{t})$ is an odd cell; mass will be added to it.

Suppose no such $m_{t+1}$ exists. Then \citet{BB02} (cf. pp.~678--679) show that all men encountered so far in the chain are man-optimal; add them to $M^{*}$ and return to Step 1.
\item If $m_{t+1}$ exists, let $w_{t+1}$ be his worst realized partner in $Q$. Since $(m_{t+1},w_{t})$ is not man-stable, $m_{t+1}$ must have positive mass matched with someone strictly worse than $w_{t}$. Hence
\[
w_{t}\succ_{m_{t+1}}w_{t+1},\quad Q_{m_{t+1},w_{t+1}}>0.
\]
Call $(m_{t+1},w_{t+1})$ the next even cell.
\item Continue alternating even and odd cells until a man or woman repeats. Then there is a simple cycle
\[
(m_{1},w_{1}),(m_{2},w_{1}),(m_{2},w_{2}),\ldots(m_{K},w_{K}),(m_{1},w_{K})
\]
where moving mass along $(m_{k},w_{k})\rightarrow(m_{k+1},w_{k})$ improves the man side.

Formally, let
\[
D=\sum_{k=1}^{K}(e_{m_{k+1},w_{k}}-e_{m_{k},w_{k}})
\]
where $e_{m,w}$ is the basis vector in $\mathcal{Q}$ denoting a 1 in the $(m,w)$ entry. Let $\bar{\alpha}=\min_{k\leq K}Q_{m_{k},w_{k}}$, and set
\[
Q^{\alpha}=Q+\alpha D,\quad0<\alpha<\bar{\alpha}.
\]

$\{Q^{\alpha}\}$ is a set of man-side improving matches.
\end{enumerate}
\citet{BB02} also show that the constructed $Q^{\alpha}$ are stable.\footnote{cf. pp.~678--679} Concisely, observe that $\{Q^{\alpha}\}$ all share the same support and thus have the same support graph $G(Q^{\alpha})$. Observe also that $\Gamma$ is a cycle in $G(Q^{\alpha})$. This shows the converse of Lemma \ref{lem:acyclic}.
\begin{lem}
\label{lem:cycle}
Let $|f(p)|>1$. Then there exists $Q^{\alpha}\in f(p)$ with support graph $G(Q^{\alpha})$ containing a cycle
\[
\Gamma=\{(m_{1},w_{1}),(m_{2},w_{1}),(m_{2},w_{2}),\ldots(m_{K},w_{K}),(m_{1},w_{K})\}.
\]
\end{lem}
\begin{proof}
Omitted.
\end{proof}
The next example continues the previous one and illustrates the construction.
\begin{example}
Continuing Example \ref{exa:2}, consider the set of stable matches $\{Q^{t}:0\leq t\leq1\}$. For $t\in(0,1)$, the support graph consists of every edge $M\times W=\{(m_{1},w_{1}),(m_{2},w_{1}),(m_{2},w_{2}),(m_{1},w_{2})\}$. These edges also form a cycle (in the order listed).
\end{example}
We next show that the component containing $\slashed{o}$ in any support graph is a tree.
\begin{lem}
\label{lem:tree}
Let $Q\in f(p)$. The connected component of the graph $G(Q)$ containing $\slashed{o}$ is a tree.
\end{lem}
\begin{proof}
Let $c_{m}\in W\cup\slashed{o}$ denote $m$'s worst realized match, and analogously for $c_{w}$. Since $Q$ is stable, every edge $(m,w)\in\supp(Q)$ has either $w=c_{m}$ or $m=c_{w}$. That is, edges are of the form $(m,c_{m})$ or $(c_{w},w)$.

Each vertex $i\in M\cup W$ contributes the edge $(i,c_{i})$ (re-ordering for $M$ or $W$ belonging), since $c$ is binding. The vertex $\slashed{o}$ contributes no additional edges beyond these. A connected component with nodes $V\ni\slashed{o}$ thus contains at most $|V|-1$ distinct edges. A connected component on $|V|$ vertices with $|V|-1$ or fewer edges is a tree.
\end{proof}

We now turn to the proof of the main result; the sketch is as follows. We start with $Q^{\alpha}$. The cycle $\Gamma$ must be part of a connected component of $G(Q^{\alpha})$ that does not contain $\slashed{o}$. Thus this component has equal mass of both sides. If we introduce infinitesimal mass to the $M$ side, the resulting stable matches cannot lie in $\supp(Q^{\alpha})$. Thus we find a discontinuous change in $f(p)$.
\begin{proof}[Proof of Theorem \ref{thm:lhc}]
Let $f(p)$ be multivalued. Then fix $Q^{\alpha}$ from Lemma \ref{lem:cycle} and its support graph $G(Q)$, which contains a cycle $\Gamma=\{(m_{k},w_{k}),(m_{k+1},w_{k})\}_{k=1,\ldots,K}$. Note that each pair $(m,w)\in\Gamma$ has $Q_{m,w}^{\alpha}>0$.

Let $p^{n}=p+\frac{1}{n}e_{m_{1}}$ (that is, adding shrinking mass to the first man in the cycle), so $p^{n}\downarrow p$. Toward a contradiction, suppose that $f$ is lhc at $p$, so there exists a subsequence of $p^{n}$ (abusing notation, denoted $n$) and sequence $Q^{n}\in f(p^{n})$ with $Q^{n}\rightarrow Q^{\alpha}$. Let $c^{n}$ be the profile of worst realized matches in $Q^{n}$.

For large enough $n$, we have $Q_{m,w}^{n}>0$ whenever $(m,w)\in\supp(Q^{\alpha})$, in particular for $(m,w)\in\Gamma$. Thus each of these edges is in $G(Q^{n})$ for large enough $n$; recall that they form a cycle. Denote the component containing this cycle as $M^{n}\cup W^{n}$; since it contains a cycle, it does not include $\slashed{o}$ by Lemma \ref{lem:tree}. Thus
\[
Q_{M^{n},W^{n}}^{n}=\sum_{m\in M^{n}}p_{m}^{n}=\sum_{w\in W^{n}}p_{w}^{n}
\]
as this component is fully matched.

Since $m_{1}\in M^{n}$, we have
\[
\left(\sum_{m\in M^{n}}p_{m}\right)+\frac{1}{n}=\sum_{w\in W^{n}}p_{w}
\]
\[
\Longrightarrow\frac{1}{n}=\sum_{w\in W^{n}}p_{w}-\sum_{m\in M^{n}}p_{m}.
\]
There are finitely many possible pairs $(M^{n},W^{n})$, so $\sum_{w\in W^{n}}p_{w}-\sum_{m\in M^{n}}p_{m}$ has finitely many possible values. (Note $p_{w},p_{m}$ are not indexed by $n$.) But this contradicts that $\frac{1}{n}\rightarrow0$.
\end{proof}

\subsection{Uniqueness}
\label{subsec:Uniqueness}

\citet{BB02} show the following sufficient condition for $|f(p)|=1$, which they term \textbf{strong nondegeneracy}.
\begin{lem}[Strong nondegeneracy. \citealp{BB02}, Theorem 9]
Fix $M$, $W$, and $p$. Suppose that for all nonempty subsets $A\subseteq M$ and $B\subseteq W$,
\[
p_{A}\neq p_{B}.
\]
Then $f(p)$ is single-valued.
\end{lem}

Thus, in Euclidean space, uniqueness is the generic case. Formally, the set of strongly nondegenerate vectors is open and dense in $\mathcal{P}$, and its complement has Lebesgue measure zero.
\begin{prop}
Define
\[
\mathcal{P}^{\mathrm{SND}}:=\left\{p\in\mathcal{P}:p_{A}\neq p_{B}\text{ for all nonempty }A\subseteq M,\ B\subseteq W\right\}.
\]
Then $\mathcal{P}^{\mathrm{SND}}$ is open and dense in $\mathcal{P}$, and $\mathcal{P}\setminus\mathcal{P}^{\mathrm{SND}}$ has Lebesgue measure zero.
\end{prop}
\begin{proof}
For any nonempty $A\subseteq M$ and $B\subseteq W$, define
\[
H_{A,B}:=\left\{p\in\mathcal{P}:p_{A}-p_{B}=0\right\}.
\]
The set $H_{A,B}$ is the intersection of $\mathcal{P}$ with a proper hyperplane, so it is closed relative to $\mathcal{P}$ and has Lebesgue measure zero. Observe that
\[
\mathcal{P}\setminus\mathcal{P}^{\mathrm{SND}}=\bigcup_{A,B}H_{A,B},
\]
where the union is finite because there are only finitely many possible pairs $(A,B)$. Hence $\mathcal{P}\setminus\mathcal{P}^{\mathrm{SND}}$ has Lebesgue measure zero. Moreover, as a finite union of relatively closed sets with empty interior, it is relatively closed and nowhere dense in $\mathcal{P}$. Therefore its complement $\mathcal{P}^{\mathrm{SND}}$ is open and dense in $\mathcal{P}$.
\end{proof}

\section{Approximation of discrete markets by continuums}

Continuum matching models are often used as approximations to large finite markets, such as in \citet{AL16} and \citet{CKK19}. In our setting, the connection is straightforward. A discrete market with many agents of finitely many types can be represented by a vector $k$ of type counts. After rescaling individuals to have vanishing mass, this becomes a population vector $p$, and the approximation is governed directly by the properties of $f(p)$.

Two conclusions arise from our continuity results. Upper hemicontinuity gives \textbf{consistency}: any limit of stable outcomes of discrete markets is stable in the limiting continuum market. Uniqueness (and thus lower hemicontinuity) gives \textbf{robustness}: if $f(p)=\{Q\}$, then every stable outcome of every sequence of discrete markets converges to $Q$. Conversely, if $f(p)$ is multi-valued, then there is some $Q\in f(p)$ and some sequence of discrete markets for which no stable outcomes approach $Q$.

Thus our results highlight two issues when choosing approximations via continuums. First is the selection of the population vector $p$---since $f(p)$ is not generally continuous, small changes in $p$ can result in discontinuous changes in the set of stable matches. Second is the exact nature of the approximation. Selection of the actual sequence of discrete markets can change the available limiting points. In words, the manner in which the market grows and agents shrink affects the available limiting points. When $p$ is such that there is a unique stable matching, both issues are resolved.

We introduce notation to formalize the above notions. Fix agent types and preferences $(M,W,\succeq)$. Let $k^{n}\in\mathbb{Z}_{++}^{M\cup W}$ be a sequence of discrete counts, where $k_{i}$ is the count of agent type $i\in M\cup W$. Let $X^{n}$ be an aggregate count matrix, where $X_{m,w}^{n}$ represents the count of type $m$ matched to $w$. Stability is defined analogously to the notion for continuum models. Formally, $X^{n}$ is stable if\footnote{To see this is nonempty, observe that $f(k^{n})$ is nonempty. Let $Q\in f(k^{n})$; then any $X^{n}$ such that $\supp(X^{n})\subseteq\supp(Q)$ is also stable. The only other requirement is for $X^{n}$ to have integral entries, which is always possible, perhaps shrinking the support.}
\[
X^{n}\in f(k^{n})\cap\mathbb{Z}_{+}^{(M\cup\slashed{o})\times(W\cup\slashed{o})}.
\]
To represent agents of shrinking mass, let $a_{n}\rightarrow\infty$ and define
\[
p^{n}:=\frac{k^{n}}{a^{n}},\quad Q^{n}=\frac{X^{n}}{a^{n}}
\]
By Lemma \ref{lem:homog}, $X^{n}\in f(k^{n})$ implies that $Q^{n}\in f(p^{n})$. Letting $k^{n},a^{n}\rightarrow\infty$ such that $p^{n}$ converges thus creates a sequence where discrete agents shrink in mass toward a continuum model.

The first result, consistency, guarantees that the limit of any sequence of stable matches $X^{n}$ is itself stable.
\begin{prop}[Consistency]
Fix $(M,W,\succeq,p)$. Let $a^{n}\rightarrow\infty$ and $k^{n}/a^{n}\rightarrow p$. Let $X^{n}$ be a convergent sequence, where each $X^{n}$ is a stable aggregate count matrix of the $n$th market. Then $X^{n}/a^{n}:=Q^{n}\rightarrow Q\in f(p)$.
\end{prop}
\begin{proof}
This is a consequence of uhc of $f(\cdot)$, by Theorem \ref{thm:uhc}. We know each $Q^{n}\in f(p^{n})$. Since $Q^{n}\rightarrow Q$ is convergent and $f(\cdot)$ is uhc, $Q\in f(p)$.
\end{proof}
The next result is a similar consequence of uhc; if $f(p)$ is a singleton, any sequence $X^{n}$ will converge.
\begin{prop}[Robustness I]
Fix $(M,W,\succeq,p)$, and suppose $f(p)=\{Q\}$. Let $a^{n}\rightarrow\infty$ and $k^{n}/a^{n}\rightarrow p$, and let $X^{n}$ be any stable aggregate count matrix of the $n$th market. Then
\[
\frac{X^{n}}{a^{n}}\rightarrow Q.
\]
\end{prop}
\begin{proof}
Since $Q^{n}:=X^{n}/a^{n}\in f(p^{n})$, the sequence is bounded. Since $f(\cdot)$ is uhc, every convergent subsequence from $f(p^{n})$ has a limit in $f(p)=\{Q\}$. Thus $Q^{n}\rightarrow Q$ as desired.
\end{proof}
The next result is the most substantial. Fix a sequence $k^{n}$. If $f(\cdot)$ is not lhc at $p$, then there will exist $Q\in f(p)$ such that no sequence $X^{n}$ will converge to it. Informally, there are continuum stable matches that are not the approximation of any discrete stable matches.
\begin{prop}[Robustness II]
Fix $(M,W,\succeq,p)$, and suppose $|f(p)|>1$. Then there exist $Q\in f(p)$, $\varepsilon>0$, and sequences $k^{n},a^{n}$ with $k^{n}/a^{n}\rightarrow p$, such that for every stable aggregate count matrix $X^{n}$ of the $n$th market,
\[
\left\Vert \frac{X^{n}}{a^{n}}-Q\right\Vert \geq\varepsilon.
\]
\end{prop}
\begin{proof}
Toward a contradiction, suppose $f(p)$ is multivalued. By Theorem \ref{thm:lhc}, $f(\cdot)$ is not lhc at $p$. Thus there exists some $Q\in f(p)$ and $\varepsilon>0$, with a sequence $x^{n}\rightarrow p$ such that
\[
f(x^{n})\cap B_{2\varepsilon}(Q)=\emptyset
\]
for infinitely many $n$.

Since $f(\cdot)$ is uhc and compact-valued, for each $n$ there exists a neighborhood $U^{n}$ of $x^{n}$ such that
\[
f(y)\cap B_{\varepsilon}(Q)=\emptyset\quad\forall y\in U^{n}.
\]
Suppose not. Then there exists a sequence $y^{\ell}\rightarrow x^{n}$ and
\[
Y^{\ell}\in f(y^{\ell})\cap B_{\varepsilon}(Q)
\]
The sequence $Y^{\ell}$ is bounded, so there exists a convergent subsequence, denoted $Y^{m}\rightarrow\bar{Y}$. Uhc gives $\bar{Y}\in f(x^{n})$ and
\[
\left\Vert \bar{Y}-Q\right\Vert \leq\varepsilon,
\]
a contradiction.

We turn to constructing the sequences $k^{n}$ and $p^{n}$. Since rational vectors are dense in $\mathcal{P}$, let
\[
p^{n}\in U^{n}\cap\mathbb{Q}_{++}^{M\cup W}
\]
such that $\left\Vert p^{n}-x^{n}\right\Vert <1/n$. Then $p^{n}\rightarrow p$. Write $p^{n}=b^{n}/c^{n}$ for some $b^{n}\in\mathbb{Z}_{++}^{M\cup W}$ and $c^{n}\in\mathbb{Z}_{++}$. Define
\[
a_{n}:=nc_{n},\quad k^{n}=nb^{n},
\]
so $a_{n}\rightarrow\infty$ and $k^{n}$ is integral, with $k^{n}/a^{n}=p^{n}\rightarrow p$.

Let $X^{n}$ be any stable aggregate count matrix of the discrete market $(M,W,\succeq,k^{n})$. Then $Q^{n}:=X^{n}/a^{n}\in f(p^{n})$ by construction. But $p^{n}\in U^{n}$, so
\[
f(p^{n})\cap B_{\varepsilon}(Q)=\emptyset
\]
Then we have
\[
\left\Vert Q^{n}-Q\right\Vert \geq\varepsilon
\]
as desired; note that $X^{n}$ was arbitrary.
\end{proof}
Intuitively, uhc guarantees that limits of the discrete market stable matches are stable. However, there may be continuum stable aggregate matches that are not approximable by discrete markets. The following example also notes that side-optimal selection may not remove the problem.
\begin{example}
Continue from Example \ref{exa:2}. Recall that for $p=(1,1,1,1)$, the stable matches are given by \[
Q^{t} =
\begin{bNiceMatrix}[first-row,first-col]
            & w_1 & w_2 & \varnothing \\
m_1         & t   & 1-t   & 0        \\
m_2         & 1-t  & t   & 0          \\
\varnothing & 0   & 0   & 0
\end{bNiceMatrix}
\]
for $t\in[0,1]$, with man-optimal stable match \[
Q^{M} =
\begin{bNiceMatrix}[first-row,first-col]
            & w_1 & w_2 & \varnothing \\
m_1         & 1   & 0   & 0        \\
m_2         & 0  & 1   & 0          \\
\varnothing & 0   & 0   & 0
\end{bNiceMatrix}.
\]
Now consider $p^{n}=(1+1/n,1,1,1)$. (It is straightforward to write discrete markets and $a^{n}$ corresponding to $p^{n}$.) As in Example \ref{exa:2}, the unique stable matching is \[
\begin{bNiceMatrix}[first-row,first-col]
            & w_1 & w_2 & \varnothing \\
m_1         & 0   & 1   & 1/n        \\
m_2         & 1   & 0   & 0          \\
\varnothing & 0   & 0   & 0
\end{bNiceMatrix}.
\]
Then $Q^{n}\rightarrow Q^{W}$ only, and the limiting market has many stable matches that are not approximable by discrete markets of this form.
\end{example}

\section{Conclusion}

We study a continuum stable marriage market. Building on the classic \citet{GS62} model, the additional topological structure allows us to derive continuity results of the correspondence from population vectors $p$ to the set of stable matches $f(p)$. The correspondence is upper hemicontinuous everywhere, but lower hemicontinuous at $p$ if and only if $f(p)$ is single-valued.

When $p$ is composed of ones as in the classic discrete setting, $f(p)$ is generally multi-valued, offering formal explanation for the chaotic nature of the set of stable matchings. However, when $p$ is drawn from real numbers, uniqueness and thus continuity is the typical case.

\end{document}